\documentclass[conference]{IEEEtran}
\IEEEoverridecommandlockouts
\usepackage{cite,amsmath,amssymb,amsthm,amsfonts,algorithmic,
graphicx,textcomp,xcolor,epsfig,mathrsfs,balance,enumitem,url}
\usepackage[ruled,linesnumbered]{algorithm2e}
\setlist{leftmargin=3.1mm}
\SetKwInput{kwInit}{Initialize}
\SetKwInput{kwUpdate}{Update}
\newtheorem{theorem}{Theorem}
\newtheorem{lemma}{Lemma}
\def\BibTeX{{\rm B\kern-.05em{\sc i\kern-.025em b}\kern-.08em
    T\kern-.1667em\lower.7ex\hbox{E}\kern-.125emX}}

\begin{document}

\title{Age of Information Analysis for \\Dynamic Spectrum Sharing\vspace{-3mm}}

\author{\IEEEauthorblockN{Yao Zhao\IEEEauthorrefmark{1},
Bo Zhou\IEEEauthorrefmark{2}, Walid Saad\IEEEauthorrefmark{2},
and Xiliang Luo\IEEEauthorrefmark{1} \\}
\IEEEauthorblockA{
\IEEEauthorrefmark{1}School of Information Science and Technology, ShanghaiTech University, Shanghai, China, \\
Email: {\tt \{zhaoyao1,luoxl\}@shanghaitech.edu.cn}\\
\IEEEauthorrefmark{2}Wireless@VT, Bradley Department of Electrical and Computer Engineering, Virginia Tech, Blacksburg, VA, USA,\\
Email: {\tt \{ecebo,walids\}@vt.edu}
\thanks{This research was supported in part by the National Natural Science Foundation of China under Grant No. 61971286, and the U.S. National Science Foundation under Grant CNS-1836802, and the Office of Naval Research under MURI Grant 1000011132.}
}\vspace{-10mm}}

\maketitle

\begin{abstract}
Timely information updates are critical to time-sensitive applications in networked monitoring and control systems. In this paper, the problem of real-time status update is considered for a cognitive radio network (CRN), in which the secondary user (SU) can relay the status packets from the primary user (PU) to the destination. In the considered CRN, the SU has opportunities to access the spectrum owned by the PU to send its own status packets to the destination. The freshness of information is measured by the age of information (AoI) metric. The problem of minimizing the average AoI and energy consumption by developing new optimal status update and packet relaying schemes for the SU is addressed under an average AoI constraint for the PU. This problem is formulated as a constrained Markov decision process (CMDP). The monotonic and decomposable properties of the value function are characterized and then used to show that the optimal update and relaying policy is threshold-based with respect to the AoI of the SU. These structures reveal a tradeoff between the AoI of the SU and the energy consumption as well as between the AoI of the SU and the AoI of the PU. An asymptotically optimal algorithm is proposed.
Numerical results are then used to show the effectiveness of the proposed policy.
\end{abstract}

\vspace{-2.0mm}
\section{Introduction}
\vspace{-2.0mm}
Maintaining timely and fresh information updates is an essential part of many emerging Internet of Things (IoT) applications, such as vehicular networks and cyber-physical system monitoring \cite{Kim2012,Walid2019}. To quantify such information freshness, the concept of \emph{age of information (AoI)} was recently proposed \cite{Kaul2012}. In essence, the AoI is defined as the elasped time since the most recent received update was generated at an end-device. 

Many recent works have addressed pertinent AoI challenges \cite{Kaul2012,Wu2018,Wang2019,Zhou2018,Hsu2017,Ceran2019,Bedewy2019,Yates2019}. For instance, the authors in \cite{Wu2018} derived the optimal status update policies to minimize the average AoI for an energy harvesting sensor equipped with a battery. The work in \cite{Wang2019} studied the optimal packet preempting policy that can be used to minimize the average AoI with rate-limited links. In \cite{Zhou2018}, the authors considered the average AoI minimization problem in a real-time IoT monitoring system, in which the IoT device incurred costs for sampling physical processes and updating its packets. The work in \cite{Hsu2017} investigated the user scheduling problem in a wireless broadcast network where massive users sent timely updates to a common destination through a shared channel. Transmission scheduling for AoI minimization under different automatic repeat request (ARQ) mechanisms was discussed in \cite{Ceran2019}. Different from the aforementioned works \cite{Wu2018,Wang2019,Zhou2018,Hsu2017,Ceran2019}, where the packets would be transmitted immediately when they are generated at the sources, there are also some works \cite{Bedewy2019,Yates2019} that consider scenarios in which the arriving packets at the sources are queued before being transmitted. It was shown that a last-come first-serve (LCFS) policy achieves age-optimality for multi-server networks \cite{Bedewy2019}. The status update system with multiple sources was discussed for different queueing patterns, LCFS and first-come first-serve (FCFS), respectively in \cite{Yates2019}. However, none of the aforementioned works have considered the problem of AoI and status updates for cognitive radio networks (CRNs), in which secondary users (SUs) must share the spectrum of the licensed, primary users (PUs).

The main contribution of this work lies in the characterization
of the optimal update and relaying policy which jointly minimizes the average AoI of the SU and its energy consumption under a constraint of on AoI performance requirement of the PU in a CRN. We consider a CRN in which the PU and the SU need to update the status information packet of their associated physical processes to a common destination. In the spectrum overlay protocol, the SU can opportunistically access the spectrum of the PU \cite{Le2008}. The SU can relay the packets generated at the PU in exchange for opportunities to forward its own packets to the destination through the spectrum owned by the PU. Such a consideration is different from the work in \cite{Gu2019}, in which no relaying procedure is involved but the PU and the SU may generate interference to each other. Since the relaying procedure will incur an additional AoI cost to the PU, one must carefully design the status update and relaying policy for the SU, such that the PU's AoI requirement is satisfied. To capture the impact that the SU's policy has on the PU's AoI, we formulate the joint AoI and energy consumption minimization problem as a constrained Markov decision process (CMDP).
We characterize two different threshold-type structures for the optimal updating and relaying policy of the SU. We show that the AoI of the PU affects the policy of the SU only when there is a packet arriving at the PU. These structures are further verified through numerical results. 

\vspace{-1.5mm}
\section{System Model}
\vspace{-1.0mm}

We consider a CRN consisting of one PU and one SU. Both the PU and the SU monitor their corresponding time-varying physical processes and transmit the update packets on the status of their physical processes to a common destination timely over a shared noiseless channel. In our model, the PU and the SU will not transmit their packets to the destination simultaneously. Hereinafter, we use the terms ``primary packet'' and ``secondary packet'' to refer to the packet originating from the PU and the SU, respectively.
We consider a time-slotted system with the slots indexed by $t=0,1,\cdots$. The transmission time of one packet is assumed to be one slot. 

We assume that the packet arrival at the PU follows a Bernoulli distribution with the parameter $p$. We define a variable $\Lambda_p(t)=1$ such that if a primary packet arrives at the PU at the beginning of time slot $t$, and $\Lambda_p(t)=0$ otherwise. Thus $\mathbb{P}(\Lambda_p(t)=1)=p$. Since there are no buffers, the primary packet would be immediately transmitted upon arrival. The SU needs to cognitively make a decision on whether to occupy the shared channel at the beginning of each time slot. When there is no primary packet arriving at the PU, the SU can occupy the shared channel freely and update the status of its physical process on demand. If the SU wants to generate and transmit a secondary packet to the destination upon the arrival of a primary packet, it will have to receive the primary packet from the PU first and then forward it to the destination in the next time slot. Obviously, it would take two time slots to the destination if the primary packet is relayed by the SU. Therefore, we need to carefully design the updating process of the SU to improve its AoI performance, while still guaranteeing the AoI performance of the PU.

\subsection{Age of Information Model}

We use the AoI metric to measure the information timeliness from the perspective of the destination. The AoI of the PU and the SU at time slot $t$ will be given by $A_p(t)$ and $A_s(t)$, respectively. Then, the evolution of the AoI of the PU can be expressed as follows.
\begin{equation}
    A_p(t+1)=
    \begin{cases}
     1,
     \begin{array}{l}
        \mbox{if the primary packet is directly}\\
        \mbox{transmitted to the destination,}
     \end{array}\\
     2,
     \begin{array}{l}
        \mbox{if the primary packet is relayed}\\
        \mbox{to the destination,}
     \end{array}\\
     A_p(t)+1, \mbox{otherwise.}
    \end{cases}
\end{equation}

The AoI would decrease to one once a primary packet is transmitted directly to the destination, because it takes one slot to complete the transmission, and two if transmitted by the SU due to the additional slot needed for relaying as shown in Fig. 1(a). Similarly, we can write the evolution of the SU's AoI as follows.
\begin{equation}
    A_s(t+1)\hspace{-1mm}=\hspace{-1mm}
    \begin{cases}
     1, \mbox{ if a secondary packet is transmitted},\\
     A_s(t)+1, \mbox{ otherwise.}
    \end{cases}\label{ageDynamic}
\end{equation}

\begin{figure}[t]
    \begin{minipage}{0.45\linewidth}
    \centering
    \centerline{\includegraphics[scale=0.22]{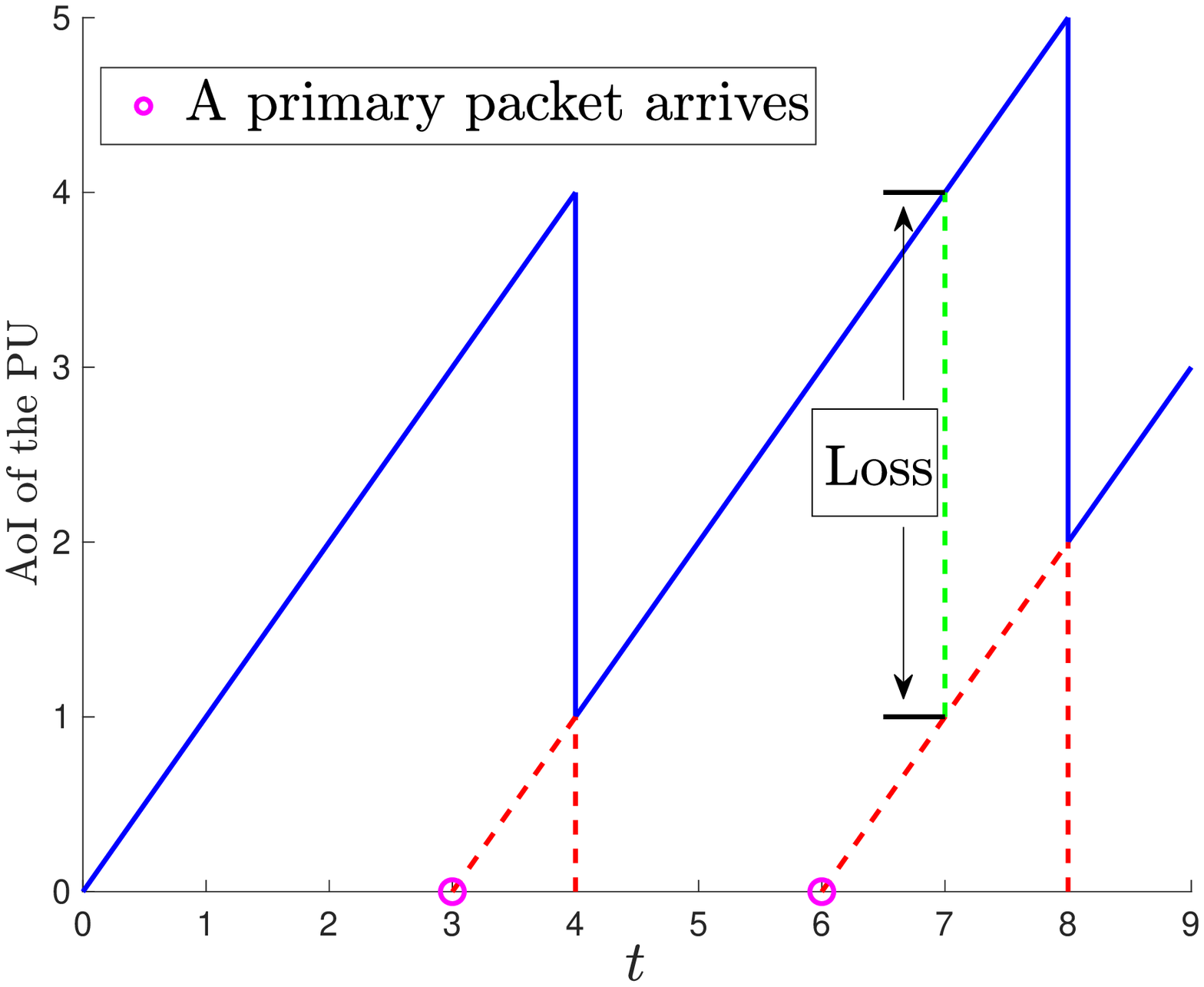}}
    \centerline{(a)}
    \medskip
    \end{minipage}
    \hfill
    \begin{minipage}{0.49\linewidth}
    \centering
    \centerline{\includegraphics[scale=0.22]{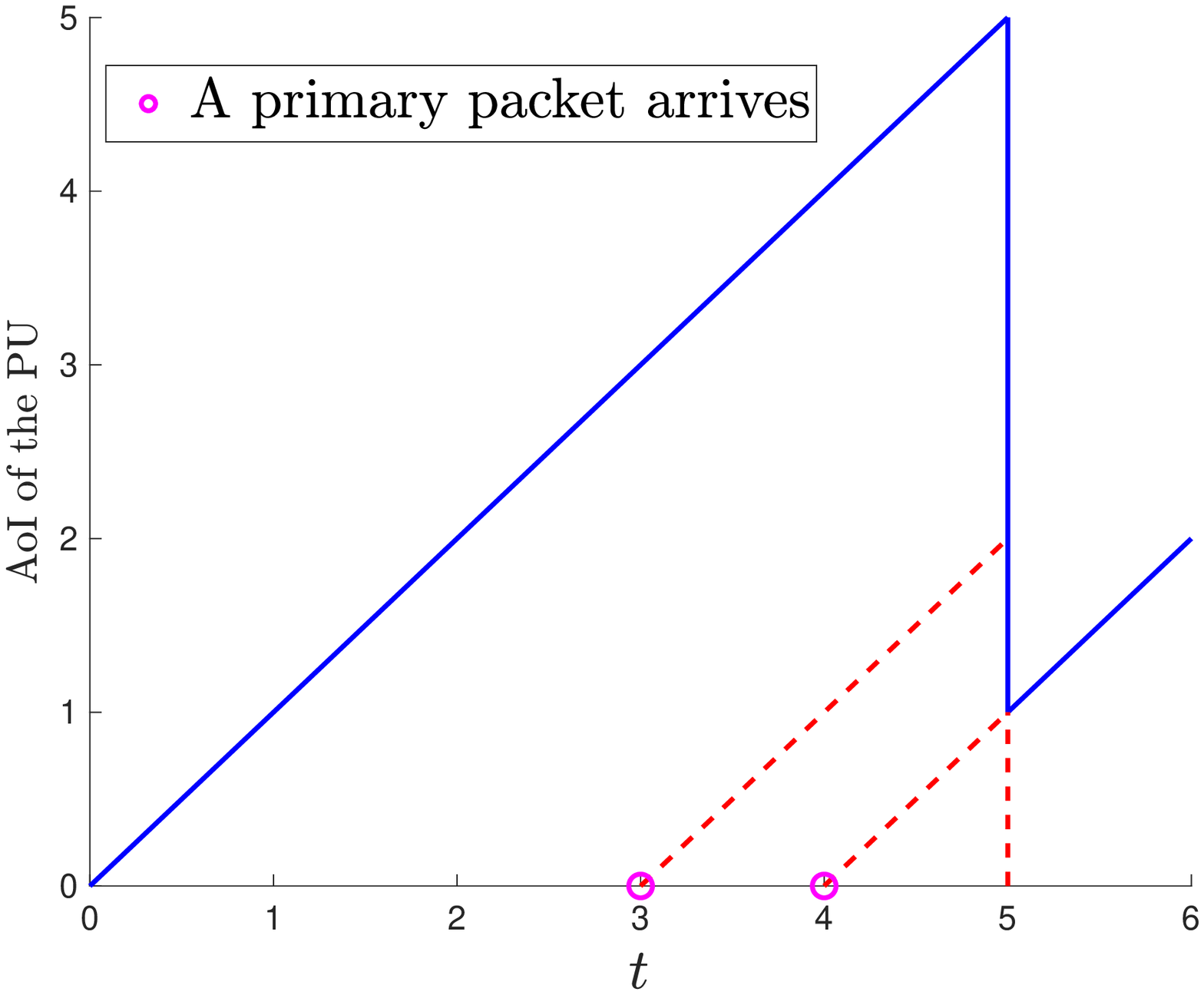}}
    \centerline{(b)}
    \medskip
    \end{minipage}
\caption{Two AoI evolution traces for the PU.
(a): The red circle means a primary packet arrives. The first one is directly transmitted to the destination by the PU itself, the second one is relayed to the SU then transmitted to the destination by the SU, the green line indicates the additional loss caused by relaying;
(b): Two primary packets arrive in succession. The first one is relayed to the SU but gets discarded since the later one is more timely.}
\vspace{-3mm}
\end{figure}

\section{Problem Formulation}

\subsection{CMDP Formulation}

Due to the shared channel to the destination, the AoI of the PU and the AoI of the SU are strongly correlated. If the SU decides to generate and transmit its own packet while receiving the primary packet to its buffer at time slot $t$, the PU would suffer additional AoI cost compared with transmitting to the destination directly. Thus, the policy of the SU must balance the fundamental tradeoff between its own performance and that of the PU. We will next cast this problem into a CMDP to allow an AoI performance guarantee for the PU. The components of the CMDP are given as follows.

\begin{itemize}
\item \textbf{States}: 
The system state at time slot $t$ is defined by a 4-tuple, i.e. $\boldsymbol{s}(t)=\left(A_p(t),A_s(t),\Lambda_p(t),\Lambda_s(t)\right)\in\mathcal{S}$, where $\mathcal{S}$ denotes the system state space. Here $\Lambda_s(t)=1$ means that there is a primary packet at the SU at the beginning of time slot $t$, which was transmitted from the PU in the previous time slot $t-1$. Note that $A_s(t)=1$ if $\Lambda_s(t)=1$, because the SU will relay the primary packet immediately after it transmits its own packet;
    
\item \textbf{Actions}: 
There are three types of states, each of which is associated with different action set: \textit{i)}. The first type is $\boldsymbol{s}=(A_p,A_s,0,0)$, for which there is no primary packet arriving at the PU, and the SU can freely choose to generate and transmit secondary packet or stay silent. \textit{ii)}. The second is $\boldsymbol{s}=(A_p,A_s,1,0)$. In this case, the SU can generate and transmit a secondary packet while the primary packet is transmitted to the SU, or the SU stays silent and the PU would transmit to the destination directly. \textit{iii)}, The third is $\boldsymbol{s}=(A_p,A_s,0\;\mbox{or}\;1,1)$. In this case, either the SU needs to transmit the primary packet, or the PU transmits the primary packet if there is a newly arrived primary packet and the primary packet at the SU would be discarded, because it is staler than the newly arrived primary packet at the PU, as seen in Fig. 1(b). The action taken in time slot $t$ is denoted by $a(t)$. For simplicity, we use $\mathcal{A}=\{1,2\}$ to indicate the two different actions. Action $1$ means the SU updates its status. For the second type state, the SU would be receiving the primary packet from the PU to its buffer while transmitting the secondary packet. Action $2$ means the SU keeps silent. Thus the PU would transmit a primary packet to the destination if there is any;
    
\item \textbf{Transition Probabilities}: 
Let $\mathbb{P}(\boldsymbol{s}_1,\boldsymbol{s}_2;1)$ be the probability that the state changes to $\boldsymbol{s}_2$ under taking action $1$ in state $\boldsymbol{s}_1$. The derivation of the transition probabilities is quite straightforward and mainly depends on the random process of the packet arrival at the PU. Due to space limitations, here, we only give an example of the second type states as follows,
\begin{subequations}
    \begin{align}
    &\mathbb{P}((A_p,A_s,1,0),(A_p+1,1,1,1);1)=p,\\
    &\mathbb{P}((A_p,A_s,1,0),(A_p+1,1,0,1);1)=1-p,\\
    &\mathbb{P}((A_p,A_s,1,0),(1,A_s+1,1,0);2)=p,\\
    &\mathbb{P}((A_p,A_s,1,0),(1,A_s+1,0,0);2)=1-p;
    \end{align}
\end{subequations}

\item \textbf{Cost}: 
The immediate cost function is defined as the weighted-sum of the AoI of the SU and the energy consumption:
\begin{equation}
    B(t)=A_s(t+1)+kC_ea(t),
\end{equation}
where $k$ is a constant to balance these two terms and $C_e$ is the energy cost needed to generate a secondary packet.
\end{itemize}

We can now formulate our joint AoI and energy consumption minimization problem as an infinite horizon average cost CMDP, given as follows.
\begin{subequations}\label{CMDP}
    \begin{align}
    &B^*=\min_{\pi}\limsup_{T\rightarrow+\infty}\frac{1}{T}\mathbb{E}\left[\sum_{t=1}^{T}B(t)\right],\\
    &\text{s.t.}\limsup_{T\rightarrow+\infty}\frac{1}{T}\mathbb{E}\left[\sum_{t=1}^{T} A_p(t)\left(\Lambda_p(t)\& a(t)\right)\right]\le d\label{Constraint},
    \end{align}
\end{subequations}
where the expectation is taken with respect to a certain policy, and $d$ represents the maximum tolerance of expected AoI cost to the PU per time slot. Here, $x\& y=1$($0$), only when $x=1$ and $y=1$ (otherwise). Obviously, when at least one of $\Lambda_p(t)$ and $a(t)$ equals zero, the PU will not be affected by the SU at all. When $\Lambda_p(t)=a(t)=1$, namely, a primary packet arrives and is relayed to the SU, the additional cost to the PU is exactly its current AoI. Hence constraint \eqref{Constraint} properly captures the effect caused by the policy at the SU.

\subsection{CMDP Relaxation via Lagrange Method}

To obtain the optimal policy for the CMDP in \eqref{CMDP}, we adopt the Lagrange method to transform the CMDP into a un-constrained MDP parameterized by introducing a Lagrange multiplier $\lambda$. The Lagrange cost function at time slot $t$ is defined as
\begin{equation}
    L(\boldsymbol{s}(t),a(t);\lambda)=B(t)+\lambda A_p(t)\left(\Lambda_p(t)\& a(t)\right).
\end{equation}
Accordingly, we have the unconstrained MDP with the optimal objective value denoted by $L^*_{\lambda}$, i.e. 
\begin{equation}
    L^*_{\lambda}=\min_{\pi}\limsup_{T\rightarrow+\infty}\frac{1}{T}\mathbb{E}\left[\sum_{t=1}^{T}L(\boldsymbol{s}(t),a(t);\lambda)\right].\label{uMDP}
\end{equation}
Fortunately, the relation between the optimal value of the constraint CMDP in \eqref{CMDP} and the optimal value of the relaxed unconstrained MDP in \eqref{uMDP} can be expressed as follows \cite{Sennott1993}.
\begin{equation}
    B^*=\max_{\lambda>0}L^*_\lambda-\lambda d.
\end{equation}
There is also a precise relation between the optimal policy of the CMDP and the optimal policy of the unconstrained MDP. According to \cite{Sennott1993} and focusing on the stationary policy, the following lemma holds for the CMDP with single constraint.
\begin{lemma}
    The optimal policy $\pi^*$ for the CMDP in \eqref{CMDP} can be expressed as a mixture of two deterministic stationary policies $\pi^*_{\lambda_1}$ and $\pi^*_{\lambda_2}$, namely,
    \begin{equation}
        \pi^*=\alpha\pi^*_{\lambda_1}+(1-\alpha)\pi^*_{\lambda_2},
    \end{equation}
    where $\alpha\in[0,1]$ is a randomization parameter and $\pi^*_{\lambda}$ is the optimal policy for the un-constrained MDP problem with the Lagrange multiplier $\lambda$.
\end{lemma}
In practice, $\lambda_1$ and $\lambda_2$ can be calculated through some iterative Lagrange multiplier estimation methods such as the Robbins–Monro algorithm \cite{spall2005}.

\section{Structures of Optimal Policies and Algorithm Designing}
In this section, we focus on the unconstrained MDP in \eqref{uMDP}. According to \cite{bertsekas2012}, the following lemma holds.
\begin{lemma}
    Given the Lagrange multiplier $\lambda$, the Bellman equation can be expressed as follows,
    \begin{equation}
        L^*_\lambda+V_\lambda(\boldsymbol{s})=\min_{\mathcal{A}}\left(L(\boldsymbol{s},a;\lambda)+\mathbb{E}\left[V_\lambda(\boldsymbol{s}')\right]\right),
    \end{equation}
    where $V_\lambda(\cdot)$ is the value function reflecting a relative and differential cost for each state $\boldsymbol{s}\in\mathcal{S}$, $\boldsymbol{s}'$ is the next state of $\boldsymbol{s}$.
\end{lemma}
The Bellman equation can be solved using the relative value iteration with an arbitrary but fixed reference state $\boldsymbol{s}_0$,
\begin{equation}
    V_{n+1,\lambda}(\boldsymbol{s})=\min_{\mathcal{A}}\left(L(\boldsymbol{s},a;\lambda)+\mathbb{E}\left[V_{n,\lambda}(\boldsymbol{s}')\right]\right)-V_{n,\lambda}(\boldsymbol{s}_0).
\end{equation}

\begin{lemma}\label{LemmaFunc}
    Given the Lagrange multiplier $\lambda$, the value function $V_\lambda(\boldsymbol{s})$ is non-decreasing in $A_p$ and $A_s$, and can be decoupled as:
    \begin{equation}
        V_\lambda(\boldsymbol{s})=f(A_p)+g(A_s).
    \end{equation}
    Moreover, for the third type state, $V_\lambda(A_p,1,1,1)$ is a constant value for all $A_p$, and so is $V_\lambda(A_p,1,0,1)$.
\end{lemma}
\begin{proof}
    All proofs can be found in our technical report \cite{fullVersion}.
\end{proof}

The properties of the value function can be used to show the following threshold-based structures of the optimal policy.

\begin{theorem}
The optimal policy exhibits the following threshold-based structures:

\begin{itemize}
\item 
For the first state type, i.e. $\boldsymbol{s}=(A_p,A_s,0,0)$, the optimal policy is of threshold type. Moreover, it is independent of the AoI of the PU, namely, there exists a constant $\eta$ such that
\begin{eqnarray}
\pi^*_\lambda(A_p,A_s,0,0)=
\begin{cases}
1, & \mbox{\rm if } A_s\ge \eta, \\
2, & \mbox{\rm otherwise};
\end{cases}
\end{eqnarray}
    
\item 
For the second state type, i.e., $\boldsymbol{s}=(A_p,A_s,1,0)$, if it is optimal to generate an update packet and relay the primary packet for the PU in state $\boldsymbol{s}=(A_p,A_s,1,0)$, then it is still optimal to choose this action for state $(A_p,A_s+1,1,0)$, a.k.a. a switch type structure.
\begin{eqnarray}
\pi^*_\lambda(A_p,A_s+1,1,0)=
\begin{cases}
1, &\hspace{-0.2cm} \mbox{\rm if } \pi^*_\lambda(A_p,A_s,1,0)=1, \\
2, &\hspace{-0.2cm} \mbox{\rm otherwise}.
\end{cases}
\end{eqnarray}
\end{itemize}
\end{theorem}

The above theorem reveals some useful properties for improving the conventional value iteration algorithm. Normally, we have to calculate all the state-action value functions for every state. Many of these calculations can be omitted by simply comparing the AoI of the PU and the SU with the states whose optimal action has been figured out in each iteration.

To cope with the countable infinite state space, we construct an approximate MDP with a finite number of states by truncating the AoI with a predetermined value $\delta$. The state set of the approximate MDP is denoted by $\mathcal{S_\delta}$, where $\mathcal{S_\delta}=\{\boldsymbol{s}|\boldsymbol{s}\in\mathcal{S},A_p\le\delta,A_s\le\delta\}$. The overall steps for solving the unconstrained MDP are summarized in Algorithm 1. Note that it can be guaranteed that the optimal policies of the approximate MDP and the original MDP problem are asymptotically identical when $\delta\to\infty$ \cite{Hsu2017}.

\begin{algorithm}[t]
\caption{Structured Value Iteration Algorithm}\label{algorithm}
    \KwIn{$\lambda$, $\delta$, $\epsilon$, $\boldsymbol{s}_0$.}
    \kwInit{$V_{0,\lambda}(\boldsymbol{s})=0,\forall \boldsymbol{s}\in\mathcal{S}_\delta$.}
    \While{$\exists\boldsymbol{s}\in\mathcal{S}_\delta\; \rm s.t. |V_{n+1,\lambda}(\boldsymbol{s})-V_{n,\lambda}(\boldsymbol{s})|>\epsilon$}{
        \For{$\boldsymbol{s}=(A_p,A_s,\cdot,\cdot)\in\mathcal{S}_\delta$}{
            \uIf{$\boldsymbol{s}=(A_p,A_s,0,0)$ \rm and $\exists A'_s<A_s,A'_p>0, \pi^*_\lambda(A'_p,A'_s,0,0)=1$}{
                $\pi^*_\lambda(A_p,A_s,0,0)=1$
            }
            \uElseIf{$\boldsymbol{s}=(A_p,A_s,1,0)$ \rm and $\exists A'_s<A_s, \pi^*(A_p,A'_s,1,0)=1$}{
                $\pi^*_\lambda(A_p,A_s,1,0)=1$
            }
            \Else{
            $\pi^*_\lambda(\boldsymbol{s})=\arg\min_{\mathcal{A}} L(s,a;\lambda)+\mathbb{E}[V_{n,\lambda}(\boldsymbol{s}')]$
            }
            $V_{n+1,\lambda}(\boldsymbol{s})\hspace{-1mm}=\hspace{-1mm}L(\boldsymbol{s},\hspace{-0.6mm}\pi^*_\lambda(\boldsymbol{s});\hspace{-0.6mm}\lambda)\hspace{-0.8mm}+\hspace{-0.8mm}\mathbb{E}[V_{n,\lambda}(\boldsymbol{s}')]\hspace{-1mm}-\hspace{-1mm}V_{n,\lambda}(\boldsymbol{s}_0)$
        }
    }
\end{algorithm}

\vspace{-1.3mm}
\section{Numerical Results}
\vspace{-1.3mm}
In this section, we evaluate, using numerical simulations, the performance of the proposed scheme. We first verify the structures of the optimal policy in Theorem 1. The parameters are set as follows. $\lambda=0.9$, $C_e=8$, and $p=0.5$. We take $\delta=20$ as a small example to illustrate the structures. From Fig. 2(a), we can see that, for the first type state $(A_p,A_s,0,0)$, the optimal policy has no relationship with the AoI of the PU, and the SU would simply choose to generate and transmit a secondary packet if its AoI is larger than $3$, or keep silent in this time slot to save energy. Fig. 2(b) shows the structure of the optimal policy for the second type state $(A_p,A_s,1,0)$. It is clearly shown that the threshold which determines the decision of the SU increases with the AoI of the PU. This is because action $1$ is essentially to transmit the secondary packet by sacrificing one transmission opportunity of the PU. Thus it is cost-effective only when the status information of the SU at the destination is stale enough.

\begin{figure}[t]
    \begin{minipage}{0.45\linewidth}
    \centering
    \centerline{\includegraphics[scale=0.3]{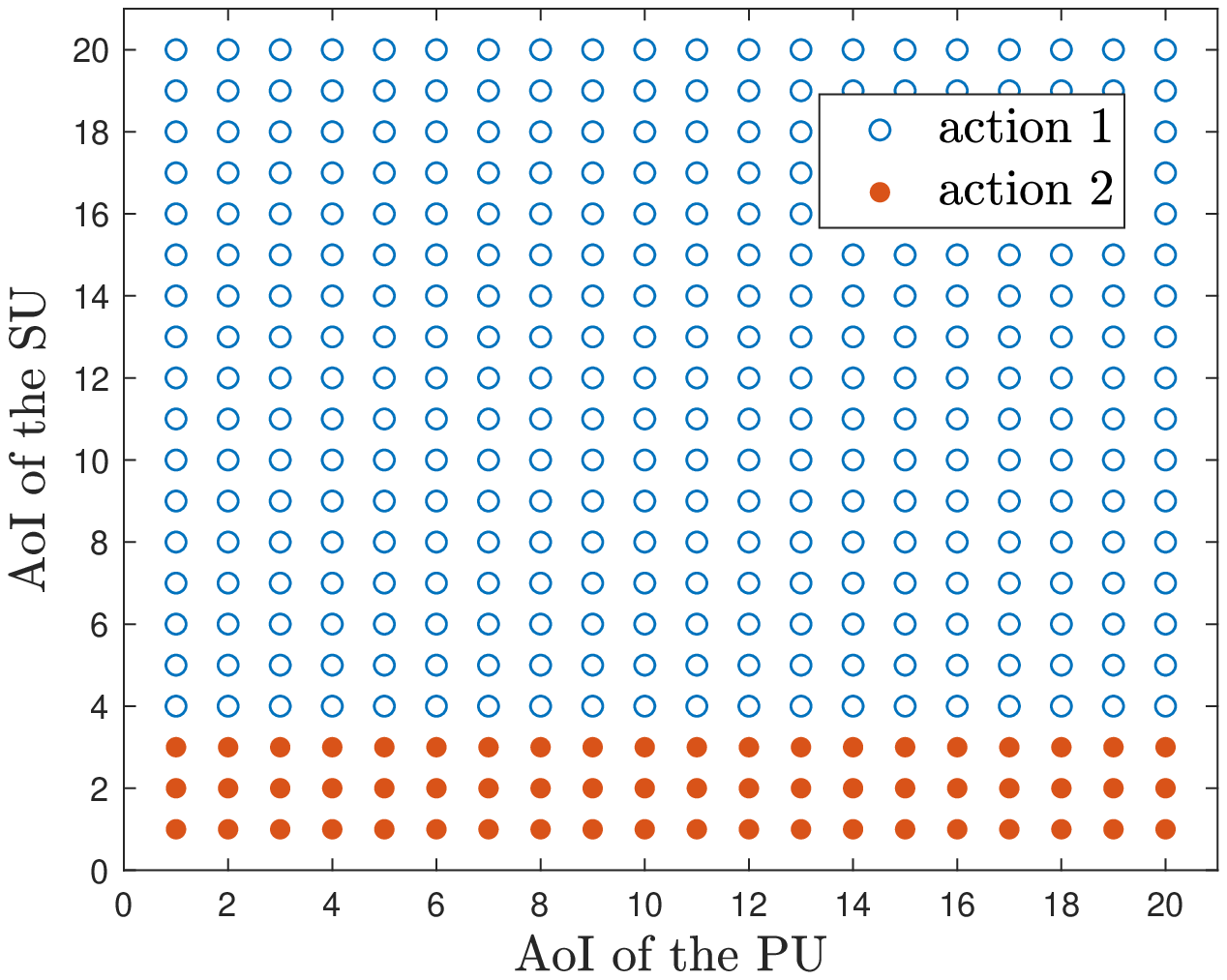}}
    \centerline{(a)}
    \medskip
    \end{minipage}
    \hfill
    \begin{minipage}{0.45\linewidth}
    \centering
    \centerline{\includegraphics[scale=0.3]{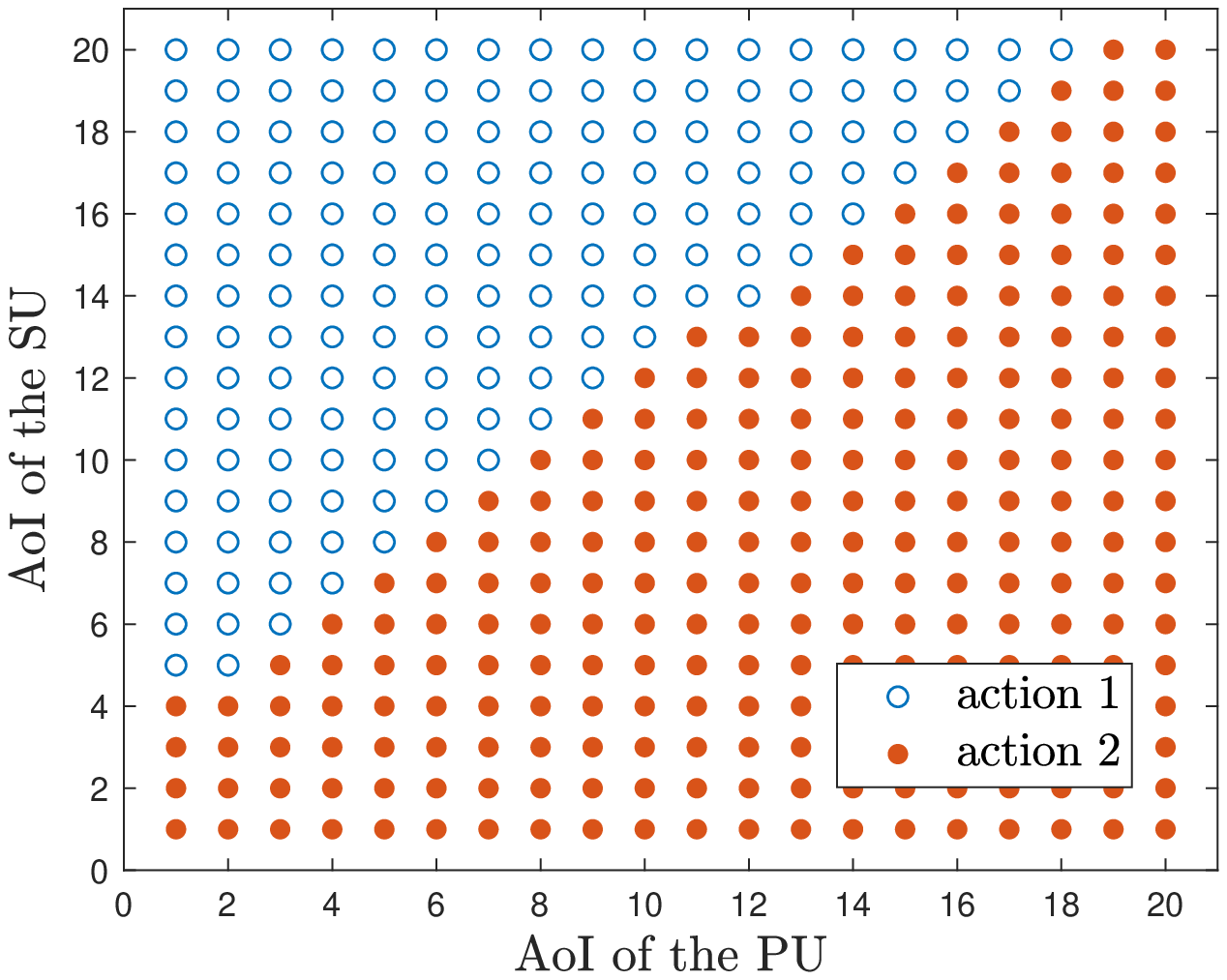}}
    \centerline{(b)}
    \medskip
    \end{minipage}
    \vspace{-2mm}
    \caption{Illustration about the different structures of the optimal policy. (a): First type state $(A_p,A_s,0,0)$; (b): Second type state $(A_p,A_s,1,0)$.}
    \vspace{-0.4cm}
\end{figure}

\begin{figure}[t]
    \epsfig{file=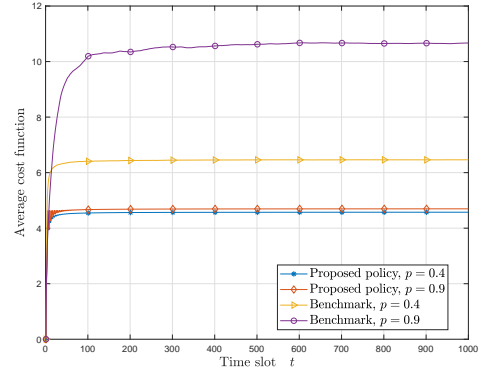,width=0.4\textwidth}
    \centering
    \vspace{-2mm}
    \caption{Convergence process of the proposed policy and the benchmark policy versus different $p$.}
    \vspace{-0.6cm}
\end{figure}


To benchmark our proposed policy, we compare it with a baseline policy. Under the baseline policy, the SU would transmit a secondary packet whenever there is no primary packet at the PU, and the SU would keep silent if there is a primary packet at the PU. This guarantees that the PU has a higher priority. Such a baseline policy is similar to the zero-waiting policy in \cite{Sun2017}. We set $\delta=200$.

As we can see in Fig. 3, the proposed policy achieves a lower cost up to $62\%$ than the baseline. Moreover, the performance of the benchmark policy is significantly affected by the parameter $p$ due to the lack of dynamically decision making compared to the proposed policy.
\vspace{-1.3mm}
\section{Conclusions}
\vspace{-1.3mm}
In this paper, we have studied the optimal update and relaying policy for the SU in a CRN. To optimize the AoI performance for both the PU and the SU, we have formulated a CMDP problem and have analyzed the structures of the optimal policy for the relaxed problem by showing the properties of the value function. It has been shown that the optimal policy is of threshold type and switch type in two different system states, respectively. We have also put forward a structure-aware value iteration algorithm by utilizing the discovered structures. Numerical simulations also show the effectiveness of the proposed policy over a benchmark policy.

\balance
\bibliographystyle{IEEEtran}
\bibliography{IEEEabrv,IEEEexample}

\begin{thebibliography}{10}
\providecommand{\url}[1]{#1}
\csname url@samestyle\endcsname
\providecommand{\newblock}{\relax}
\providecommand{\bibinfo}[2]{#2}
\providecommand{\BIBentrySTDinterwordspacing}{\spaceskip=0pt\relax}
\providecommand{\BIBentryALTinterwordstretchfactor}{4}
\providecommand{\BIBentryALTinterwordspacing}{\spaceskip=\fontdimen2\font plus
\BIBentryALTinterwordstretchfactor\fontdimen3\font minus
  \fontdimen4\font\relax}
\providecommand{\BIBforeignlanguage}[2]{{%
\expandafter\ifx\csname l@#1\endcsname\relax
\typeout{** WARNING: IEEEtran.bst: No hyphenation pattern has been}%
\typeout{** loaded for the language `#1'. Using the pattern for}%
\typeout{** the default language instead.}%
\else
\language=\csname l@#1\endcsname
\fi
#2}}
\providecommand{\BIBdecl}{\relax}
\BIBdecl

\bibitem{Kim2012}
K.~{Kim} and P.~R. {Kumar}, ``Cyber–physical systems: A perspective at the
  centennial,'' \emph{Proc. IEEE}, vol. 100, no. Special Centennial Issue, pp.
  1287--1308, May 2012.

\bibitem{Walid2019}
W.~{Saad}, M.~{Bennis}, and M.~{Chen}, ``A vision of 6{G} wireless systems:
  Applications, trends, technologies, and open research problems,'' \emph{IEEE
  Network}, 2019, to appear.

\bibitem{Kaul2012}
S.~{Kaul}, R.~{Yates}, and M.~{Gruteser}, ``Real-time status: How often should
  one update?'' in \emph{Proc. of IEEE International Conference on Computer
  Communications (INFOCOM)}, Orlando, FL, USA, Mar. 2012, pp. 2731--2735.

\bibitem{Wu2018}
X.~{Wu}, J.~{Yang}, and J.~{Wu}, ``Optimal status update for age of information
  minimization with an energy harvesting source,'' \emph{IEEE Trans. Green
  Commun. and Netw.}, vol.~2, no.~1, pp. 193--204, Mar. 2018.

\bibitem{Wang2019}
B.~{Wang}, S.~{Feng}, and J.~{Yang}, ``When to preempt? age of information
  minimization under link capacity constraint,'' \emph{IEEE J. Commun. Netw.},
  2019, to appear.

\bibitem{Zhou2018}
B.~{Zhou} and W.~{Saad}, ``Optimal sampling and updating for minimizing age of
  information in the internet of things,'' in \emph{Proc. of IEEE Global
  Communications Conference (GLOBECOM)}, Abu Dhabi, UAE, Dec. 2018, pp. 1--6.

\bibitem{Hsu2017}
Y.~{Hsu}, E.~{Modiano}, and L.~{Duan}, ``Age of information: Design and
  analysis of optimal scheduling algorithms,'' in \emph{Proc. of IEEE
  International Symposium on Information Theory (ISIT)}, Aachen, Germany, Jun.
  2017, pp. 561--565.

\bibitem{Ceran2019}
E.~T. {Ceran}, D.~{Gündüz}, and A.~{György}, ``Average age of information
  with hybrid {ARQ} under a resource constraint,'' \emph{IEEE Trans. Wireless
  Commun.}, vol.~18, no.~3, pp. 1900--1913, Mar. 2019.

\bibitem{Bedewy2019}
A.~M. {Bedewy}, Y.~{Sun}, and N.~B. {Shroff}, ``Minimizing the age of
  information through queues,'' \emph{IEEE Trans. Inf. Theory}, 2019, early
  access.

\bibitem{Yates2019}
R.~D. {Yates} and S.~K. {Kaul}, ``The age of information: Real-time status
  updating by multiple sources,'' \emph{IEEE Trans. Inf. Theory}, vol.~65,
  no.~3, pp. 1807--1827, Mar. 2019.

\bibitem{Le2008}
L.~B. {Le} and E.~{Hossain}, ``Resource allocation for spectrum underlay in
  cognitive radio networks,'' \emph{IEEE Trans. Wireless Commun.}, vol.~7,
  no.~12, pp. 5306--5315, Dec. 2008.

\bibitem{Gu2019}
Y.~Gu, H.~Chen, C.~Zhai, Y.~Li, and B.~Vucetic, ``Minimizing age of information
  in cognitive radio-based {IoT} systems: Underlay or overlay?'' \emph{arXiv
  preprint arXiv:1903.06886}, Aug. 2019.

\bibitem{Sennott1993}
L.~I. Sennott, ``Constrained average cost markov decision chains,''
  \emph{Probab. Eng. Inf. Sci.}, vol.~7, no.~1, pp. 69--83, Jan. 1993.

\bibitem{spall2005}
J.~C. {Spall}, \emph{Introduction to Stochastic Search and Optimization:
  Estimation, Simulation, and Control}.\hskip 1em plus 0.5em minus 0.4em\relax
  John Wiley \& Sons, 2005, vol.~65.

\bibitem{bertsekas2012}
D.~P. Bertsekas, \emph{Dynamic Programming and Optimal Control, 4th ed, volume
  II}.\hskip 1em plus 0.5em minus 0.4em\relax Athena scientific, 2012.

\bibitem{fullVersion}
{Y. Zhao, B. Zhou, W. Saad, and X. Luo}, ``Age of information analysis for
  dynamic spectrum sharing,'' Website, 2019,
  \url{www.dropbox.com/s/plx713rjj806e32/AoICRNs_GlobalSip2019.pdf?dl=0}.

\bibitem{Sun2017}
Y.~{Sun}, E.~{Uysal-Biyikoglu}, R.~D. {Yates}, C.~E. {Koksal}, and N.~B.
  {Shroff}, ``Update or wait: How to keep your data fresh,'' \emph{IEEE Trans.
  Inf. Theory}, vol.~63, no.~11, pp. 7492--7508, Nov. 2017.

\end{thebibliography}

\end{document}